\pdfoutput=1
\newif\ifSC
\SCfalse
\ifSC
\documentclass[12pt, draftclsnofoot, onecolumn]{IEEEtran}
\else
\documentclass[letterpaper,twocolumn,twoside]{IEEEtran}
\fi

\usepackage{comment}

\usepackage{amsmath}
\usepackage{mathtools}
\usepackage{amsmath}
\usepackage{amssymb}
\usepackage{cite}
\usepackage{array}
\usepackage{amsthm}
\usepackage{amsmath,dsfont}
\usepackage{tabulary}
\usepackage{multirow}
\usepackage{cleveref}
\usepackage{subfig}
\usepackage[font=small]{caption}
\usepackage{color}
\usepackage{pifont}
\usepackage{epsfig}
\usepackage{graphicx}
\usepackage{url}
\usepackage{amssymb}
\IEEEoverridecommandlockouts

\newtheorem{thm}{{\bf Theorem}}
\newtheorem{cor}{Corollary}

\theoremstyle{definition}
\newtheorem{definition}{Definition}

\begin{document}

\title{Correction Factor for Analysis of MIMO Wireless Networks With Highly Directional Beamforming}
\author{Mandar N. Kulkarni,  Eugene Visotsky, Jeffrey G. Andrews
\thanks{Email:
\ifSC
$\mathtt{\{mandar.kulkarni@,jandrews@ece.\}utexas.edu,eugene.visotsky@nokia-bell-labs.com}.$
\else $\mathtt{\{mandar.kulkarni@,jandrews@ece.\}utexas.edu,eugene.}$ $\mathtt{visotsky@nokia-bell-labs.com}.$
\fi M. Kulkarni and J. Andrews are with the UT Austin, TX. E. Visotsky is with Nokia Bell Labs, IL. Last revised on \today.}}
\maketitle

\begin{abstract}
In this letter, we reconsider a popular simplified received signal power model with single stream beamforming employed by the transmitter and the receiver in the regime when the beams have high gain and narrow beamwidth.  We define {\em the correction factor} as the ratio of the average actual received signal power divided by the average received signal power using the popular simplified model. We analytically quantify this factor for LOS and NLOS service and interfering links under some assumptions. The analysis along with simulations using a 3GPP compliant new radio (NR) channel model confirm the importance of incorporating the correction factor in coverage analysis of wireless networks that utilize the popular simplified received power model. 
\end{abstract}
\section{Introducing The Correction Factor}
In system level analysis for computing coverage and rate performance of wireless networks on $\mathbb{R}^2$ a popular model to compute the received signal power at $X\in\mathbb{R}^2$ from a transmitter (serving/interfering) at $Y\in\mathbb{R}^2$ is as follows  \cite{Yu03,Hun08,BaiHea14,Wildman14}.
\begin{equation}
\label{eq:singlepath}
P_r = P_t \ell(||X-Y||) h G_t(\theta) G_r(\phi),
\end{equation}
where $P_t$ is the transmit power, $\ell(.)$ is the path loss, $h$ is the small scale fading, $G_t(\theta)$ is the transmit antenna gain and $G_r(\phi)$ is the receive antenna gain. If at all blockage effects are explicitly incorporated in the analysis by differentiating line of sight (LOS) and non-LOS (NLOS) links, then only $h$ and $\ell$ are modeled differently for LOS and NLOS\cite{BaiHea14}. The antenna patterns $G_t(.)$ and $G_r(.)$ are considered to have the same distribution for LOS/NLOS links. In this work, we will show the importance of incorporating an additional blockage dependent factor in the received signal power when the antenna patterns have very narrow beamwidths and large gains -- for example, an antenna pattern having $36$ dB gain and $12^o$ half power beamwidth in azimuth. Our analytical model shows that if there are large number of antennas at the transmitter and receiver, which employ analog beamforming, then the additional factor (called as the correction factor) is much less than 1 for NLOS service links but is close to 1 for LOS service links, and equal to 1 for NLOS/LOS interfering links. Such a factor cannot be incorporated by modifying either $h$ or $\ell(.)$ for analyzing signal to interference plus noise ratio (SINR) in highly directional MIMO wireless networks, especially cellular networks, and an example to explain this is given in the Appendix. 

Most prior analyses of MIMO wireless networks computing coverage and rate performance with highly directional single user beamforming incorporates a received signal power similar to \eqref{eq:singlepath} and do not model a channel with LOS/NLOS dependent rank \cite{BaiHea14,Li17,Maamari16,Park16,Deng17,Kulkarni17,Wildman14,Hun08}\footnote{Except our prior work in \cite{KulGhoAnd16} which was the first attempt to do so, to the best of our knowledge. }, which gives rise to the needed correction factor as we will show in this work. The purpose of this letter is to make the growing research community using received power model similar to \eqref{eq:singlepath} aware of the significance of how different rank of the MIMO channel for LOS and NLOS can affect the effective antenna gain and thus the design insights. We will formally define effective antenna gain in this work. Also we propose a quick way to preserve the existing analyses by multiplication of a LOS/NLOS dependent constant for service links but not the interfering links. The example in Appendix is indicative of how this can be done. The correction factor is especially important for analysis of millimeter wave (mmWave) cellular networks, wherein inclusion of blockage effects is crucial and the beamforming is highly directional\cite{BaiHea14}. All prior works which studied different system design issues in these networks like  \cite{Li17,Maamari16,Park16,Deng17,Kulkarni17} use the received power model in \eqref{eq:singlepath} without incorporating the correction factor. In Section V, we discuss key implications on system design resulting from incorporation of such a factor. 

The analysis in this work is for analog beamforming implementation done at the transmitter and the receiver under consideration.  Our analysis along with the simulation results considering a more detailed wideband 3GPP channel model suffice to motivate the inaccuracy of the popular model in \eqref{eq:singlepath} when the transmit and receive beams are narrow and with large gains. However, more detailed analysis is needed in the future to estimate the correction factor more accurately.  

\section{System Model}
\label{sec:sysmodel}
We concentrate on a single transmitter-receiver pair in a wireless network.  $N_{t}$ and $N_r$ denote the number of transmit and receive antennas. If the link is NLOS, the narrowband channel between the transmitter-receiver is given by \cite{Akd14,3gpp_tr38901}
\begin{equation}
\label{eq:multipathchannel}
{\bf H}_\mathrm{NLOS} =  \kappa \sqrt{\frac{\ell(d)}{\eta}}\sum_{i=1}^{\eta}\gamma_i {\bf a}_{r}(\phi_i){\bf a}^*_{t}(\theta_i),
\end{equation}
where $\ell(d)$ is the path gain (assumed deterministic function of $d$ for simplicity),  
$\eta$ is the number of paths (assumed constant),  $d$ is the transmission distance in meters and $\gamma_i$ is the small scale fading on path $i$ (random variable such that $\mathbb{E}\left[|\gamma_i|^2\right] = 1$) and $\kappa$ is a normalizing constant such that $\mathbb{E}\left[||{\bf H}_\mathrm{NLOS}||^2_F\right] = N_t N_r \ell(d)$.  

Assuming a uniform linear array at the receiver, the array response vector ${\bf a}_{r}$ is given as ${\bf a}_{r}(\phi_i) = \left[1\; e^{-j \phi_i} \;e^{-2j\phi_i} \;\ldots \;e^{- (N_{r}-1)\phi_i}\right]^{T}$, where $j$ is square root of $-1$. 
Similarly, one can define ${\bf a}_t$ by replacing $N_r$ with $N_t$. Note that $\phi_i$ and $\theta_i$ are spatial angles of arrival and departure (AOA/AOD). It is assumed that these AOAs and AODs are continuous random variables and no assumption on their distribution is made. 

If the link is LOS, the narrowband channel is given by\cite{3gpp_tr38901} 
\ifSC
\begin{equation}
\label{eq:LOSmultipathchannel}
{\bf H}_\mathrm{LOS} =  \sqrt{\ell(d)} \left(\sqrt{ \frac{\mathrm{K}_\mathrm{R}}{\mathrm{K}_\mathrm{R}+1}} {\bf a}_{r}(\phi_0){\bf a}^*_{t}(\theta_0) +\right.
\left.\kappa \sqrt{\frac{1}{\eta(\mathrm{K}_\mathrm{R}+1)}}\sum_{i=1}^{\eta}\gamma_i {\bf a}_{r}(\phi_i){\bf a}^*_{t}(\theta_i)\right),
\end{equation}
\else
\begin{multline}
\label{eq:LOSmultipathchannel}
{\bf H}_\mathrm{LOS} =  \sqrt{\ell(d)} \left(\sqrt{ \frac{\mathrm{K}_\mathrm{R}}{\mathrm{K}_\mathrm{R}+1}} {\bf a}_{r}(\phi_0){\bf a}^*_{t}(\theta_0) +\right.\\
\left.\kappa \sqrt{\frac{1}{\eta(\mathrm{K}_\mathrm{R}+1)}}\sum_{i=1}^{\eta}\gamma_i {\bf a}_{r}(\phi_i){\bf a}^*_{t}(\theta_i)\right),
\end{multline}
\fi
where $\mathrm{K}_\mathrm{R}$ is the Rician K-factor. AOA and AOD given by $\phi_0$ and $\theta_0$ are constants corresponding to the direct LOS path between the receiver and the transmitter. Rest of the AOA/AOD are continuous random variables. $\eta$ and $\kappa$ could have different LOS-specific values here, as compared to \eqref{eq:multipathchannel}. 

Assuming ${\bf w}$ is the combiner employed by the receiver and ${\bf f}$ is the precoder employed by the transmitter, the received signal power model is given as 
$
P^{\mathrm{multi}}_r = ||{\bf w}^* {\bf H} {\bf f}||^2,
$
where ${\bf H}$ is the $N_r\times N_t$ channel which could be either ${\bf H}_\mathrm{LOS}$ or ${\bf H}_\mathrm{NLOS}$. We constrain  ${\bf w}$ and ${\bf f}$ to be chosen of the form $\frac{1}{N_r}{\bf a}_r(.)$ and $\frac{1}{N_t}{\bf a}_t(.)$, respectively, which is basically employing analog beamforming using phase shifters at both the receiver and the transmitter. If the transmitter-receiver pair form a desired communication link, ${\bf w}$ and ${\bf f}$ are chosen so as to maximize $P^{\mathrm{multi}}_r $. If the transmitter-receiver pair form an interfering link, then ${\bf w}$ and ${\bf f}$ can be arbitrary. 

Most analytical studies to compute coverage and rate performance cannot afford to use the received signal power model defined above for tractability. As mentioned in Section I, a simplified model similar to \eqref{eq:singlepath} is generally used. Now we will define a generative model for such a simplified model. We define a keyhole channel as follows\cite{Chizhik00}. ${\bf H}_\mathrm{keyhole} = \sqrt{\ell(d)}\gamma {\bf a}_r(\phi){\bf a}_t(\theta)$, where $\mathbb{E}\left[|\gamma|^2\right] = 1$ and $\{\theta,\phi\}$ could have arbitrary distribution. Now, $P^{\mathrm{keyhole}}_r$ is defined as $||{\bf w}^* {\bf H}_\mathrm{keyhole} {\bf f}||^2$. If the transmitter-receiver pair is a desired signal link, ${\bf w} =\frac{1}{N_r} {\bf a}_r(\phi)$ and ${\bf f} =\frac{1}{N_t} {\bf a}_t(\theta)$ to maximize $P^{\mathrm{keyhole}}_r$ and thus, $P^{\mathrm{keyhole}}_r = |\gamma|^2 \ell(d) N_t N_r$. If the transmitter-receiver pair is an interfering link with ${\bf w} = \frac{1}{N_r}{\bf a}_r(\phi')$ and ${\bf f} =\frac{1}{N_t} {\bf a}_t(\theta')$ for some arbitrary $\phi'$ and $\theta'$, then $P^{\mathrm{keyhole}}_r = \ell(d)|\gamma|^2 G_r(\phi,\phi') G_t(\theta',\theta)$, where  $G_r(\phi,\phi') = ||\frac{1}{\sqrt{N_r}}{\bf a}^*_r(\phi){\bf a}_r(\phi')||^2$. Similarly $G_t$ can be written replacing subscript $r$ with $t$ and $\phi$ with $\theta$.  Unlike the desired signal power case, here $\phi'$ and $\theta'$ are not chosen to maximize $||{\bf w}^*{\bf H}_\mathrm{keyhole}{\bf f}||^2$ but can be random angles distributed according to some continuous distribution.

We wish to compare $\mathbb{E}\left[P^{\mathrm{multi}}_r\right]$ with  $\mathbb{E}\left[P^{\mathrm{keyhole}}_r\right]$. This comparison will highlight how important it is to consider rank $>1$ channels for LOS and NLOS in terms mean received signal power since the keyhole channel is always rank 1. In order to quantify this comparison, we define a correction factor as follows. 

\begin{definition}
The proposed {\em correction factor} to estimate the received signal power on a serving/interfering link is defined as $\Upsilon = \mathbb{E}\left[P^{\mathrm{multi}}_r\right]/\mathbb{E}\left[P^{\mathrm{keyhole}}_r\right]$. 
\end{definition}

Note that for serving links $\mathbb{E}\left[P^{\mathrm{keyhole}}_r\right] = N_t N_r \ell(d)$ irrespective of LOS/NLOS as per our analytical model. 

\begin{definition}
The {\em effective antenna gain} is defined as the actual received signal power (on serving/interfering links) normalized by the path loss and the transmit power of the signal. 
\end{definition}

Note that the effective antenna gain is in general a random variable. As per our analytical model, it is equal to $P^{\mathrm{multi}}_r/\ell(d)$. Considering our system model, wherein $\ell(d)$ is deterministic the mean effective antenna gain for a serving link is given by $\Upsilon N_t N_r$, where $\Upsilon$ is the correction factor for a serving link. Our proposal is that if one wants to use a simplified received power model like in \eqref{eq:singlepath} for system level analysis, wherein the impact of beamforming is captured through a spatial gain pattern at the transmitter and receiver, then the corrected received signal power on serving and interfering links is obtained by multiplying $\Upsilon$ to the estimate in \eqref{eq:singlepath}. Since here a keyhole model is used to generate the simplified received power model in \eqref{eq:singlepath}, the corrected received signal power is $\Upsilon P^{\mathrm{keyhole}}_r$. 

\section{Computing $\Upsilon$ When $N_t, N_r$ Are Large}
\label{sec:sysmodel2}
Before we state the results, we make a quick observation based on the result in \cite{Aya12}.

{\bf Observation 1:} As $N_r \to \infty$ and $N_t \to \infty$, the left singular vectors  corresponding to non-zero singular values of \eqref{eq:multipathchannel} and \eqref{eq:LOSmultipathchannel} converge to $\frac{1}{\sqrt{N_r}}{\bf a}_r(\phi_i)$, with $i=1\ldots, \eta$ for \eqref{eq:multipathchannel} and $i=0,\ldots,\eta$ for \eqref{eq:LOSmultipathchannel}. Similarly, the right singular vectors corresponding to non-zero singular values of \eqref{eq:multipathchannel} and \eqref{eq:LOSmultipathchannel} converge to $\frac{1}{\sqrt{N_t}}{\bf a}_t(\phi_i)$. 

{\bf Observation 2:} As $N_r \to \infty$, ${\bf a}^*_r(\phi_i){\bf a}_r(\phi_j)/N_r \to \mathds{1}(i=j)$. Similarly ${\bf a}^*_t(\theta_i){\bf a}_t(\theta_j)/N_t \to \mathds{1}(i=j)$ as $N_t \to \infty$.

\begin{thm}
Large $N_t$ and $N_r$ is assumed. If the link is a NLOS service link, then
\ifSC 
$\mathbb{E}\left[P^{\mathrm{multi}}_r\right] \approx N_t N_r \ell(d)\times $ $\mathbb{E}\left[\max_{i=1,\ldots,\eta} |\gamma_i|^2\right]/\eta.$
\else
$\mathbb{E}\left[P^{\mathrm{multi}}_r\right] \approx N_t N_r \ell(d) \mathbb{E}\left[\max_{i=1,\ldots,\eta} |\gamma_i|^2\right]/\eta.$
\fi
If the link is a LOS service link and $\mathrm{K}_\mathrm{R}\gg 1$ then 
$\mathbb{E}\left[P^{\mathrm{multi}}_r\right] \approx  N_t N_r \ell(d) \mathrm{K}_\mathrm{R}/(\mathrm{K}_\mathrm{R}+1).$
\end{thm}
\begin{proof}
Optimal combiner and precoder correspond to the singular vectors corresponding to the maximum singular value norm of the channel matrix. Making use of Observation 1 for NLOS channel with large number of antennas, ${\bf w} = \frac{1}{\sqrt{N_r}}{\bf a}_{r}(\phi_1)$ and ${\bf f} = \frac{1}{\sqrt{N_t}}{\bf a}_t(\theta_1)$ assuming $|\gamma_1| = \max_i |\gamma_i|$, without loss of generality. Thus, 
\ifSC
\begin{align*}
P^{multi}_r &= ||{\bf w}^* {\bf H}_{\mathrm{NLOS}} {\bf f}||^2 =  \left|\left|\kappa \sqrt{\frac{\ell(d)}{\eta N_t N_r}}\sum_{i=1}^{\eta}\gamma_i{\bf a}^{*}_r(\phi_1) {\bf a}_r(\phi_i){\bf a}^{*}_t(\theta_i){\bf a}_t(\phi_1)\right|\right|^2.
\end{align*}
\else
\begin{align*}
P^{multi}_r &= ||{\bf w}^* {\bf H}_{\mathrm{NLOS}} {\bf f}||^2\\& =  \left|\left|\kappa \sqrt{\frac{\ell(d)}{\eta N_t N_r}}\sum_{i=1}^{\eta}\gamma_i{\bf a}^{*}_r(\phi_1) {\bf a}_r(\phi_i){\bf a}^{*}_t(\theta_i){\bf a}_t(\phi_1)\right|\right|^2.
\end{align*}
\fi
Note that Observation 1 implies that the non-zero singular values of \eqref{eq:multipathchannel} are given by $\sqrt{\frac{\ell(d) N_t N_r}{\eta}}\kappa\gamma_i$. Thus, $||{\bf H}_\mathrm{NLOS}||^2_F = \kappa^2 \ell(d) N_t N_r \sum_{i=1}^{\eta}|\gamma_i|^2/\eta$, which is computed using the fact that square of Frobenius norm equals sum of squares of singular values of a matrix. Thus $\mathbb{E}\left[||{\bf H}_\mathrm{NLOS}||^2_F\right] =  N_t N_r \ell(d) \kappa^2$, which implies that the normalizing constant $\kappa = 1$. Similarly,  $\kappa = 1$ in \eqref{eq:LOSmultipathchannel}.

Since the AODs/AOAs are continuous random variables, any two such angles are unequal with probability 1. Using the orthogonality of the array response vectors for unequal AODs/AOAs, we get $P^{\mathrm{multi}}_r  \approx \left|\left|N_r N_t\gamma_1 \sqrt{\frac{\ell(d)}{\eta N_t N_r}} + 0\right|\right|^2  = N_t N_r \ell(d) \frac{|\gamma_1|^2}{\eta}$ with probability 1. Thus, the expectation of $P^{\mathrm{multi}}_r$ is $N_t N_r \ell(d) \frac{\mathbb{E}\left[|\gamma_1|^2\right]}{\eta}$. The result is approximate as we used asymptotic results in Observations 1 and 2 for finite  number of antennas. 

For LOS, since $\mathbb{E}\left[|\gamma_i|^2\right] = 1$, by Markov inequality $\mathbb{P}\left(|\gamma_i|^2 > \eta\mathrm{K}_\mathrm{R}\right)<1/\eta \mathrm{K}_\mathrm{R}$. Thus, owing to $\mathrm{K}_\mathrm{R}\gg 1$ with high probability the maximum singular value corresponds to the direct LOS path. This implies 
that ${\bf w} = \frac{1}{\sqrt{N_r}}{\bf a}_{r}(\phi_0)$ and ${\bf f} = \frac{1}{\sqrt{N_t}}{\bf a}_t(\theta_0)$,  which are singular vectors corresponding to the maximum singular value as per Observation 1. Thus, it is concluded that 
\ifSC
\begin{align*}
P^{\mathrm{multi}}_r \approx  ||{\bf a}^*_{r}(\phi_0) {\bf H}_{\mathrm{LOS}} {\bf a}_t(\theta_0)||^2  = \left|\left|\sqrt{\ell(d) N_t N_r\mathrm{K}_\mathrm{R}/(1+\mathrm{K}_\mathrm{R})} + \rho\right|\right|^2,
\end{align*}
\else
$
P^{\mathrm{multi}}_r \approx  ||{\bf a}^*_{r}(\phi_0) {\bf H}_{\mathrm{LOS}} {\bf a}_t(\theta_0)||^2  = \left|\left|\sqrt{\ell(d) N_t N_r\mathrm{K}_\mathrm{R}/(1+\mathrm{K}_\mathrm{R})} + \rho\right|\right|^2,
$
\fi
where $$\rho = \frac{1}{\sqrt{\eta(\mathrm{K}_\mathrm{R}+1)}}\sum_{i=1}^{\eta}\gamma_i {\bf a}^{*}_r(\phi_0){\bf a}_r(\phi_i){\bf a}^{*}_t(\theta_i){\bf a}_t(\theta_0).$$
Note that using Observation 2, we have $\rho\approx 0$ by similar arguments as for NLOS case considering the angles of arrival/departure are continuous random variables. Thus, $\mathbb{E}\left[P^{\mathrm{multi}}_r\right] \approx N_t N_r \ell(d)\frac{\mathrm{K}_\mathrm{R}}{\mathrm{K}_\mathrm{R}+1}$. 

\end{proof}

\begin{cor}
Large $N_t$ and $N_r$ is assumed and the link under consideration is assumed to be a service link. If $\gamma_i$ are complex normal random variables and independent of each other, $ \Upsilon \approx \frac{1}{\eta}\sum_{k=1}^{\eta}(1/k)$ if the link is NLOS.  If $\gamma_i$ are all identical to complex normal $\gamma_1$, $ \Upsilon \approx \frac{1}{\eta}$ for NLOS link. For LOS link and $\mathrm{K}_\mathrm{R}\gg 1$, $\Upsilon \approx \frac{\mathrm{K}_\mathrm{R}}{1+\mathrm{K}_\mathrm{R}} \approx 1$.
\end{cor}
\begin{proof}
If $\gamma_i$ are complex normal random variables, $|\gamma_i|^2$ are exponentially distributed with unit mean. Also these are independent random variables. Thus,  $\mathbb{E}\left[\max_{i=1,\ldots, \eta} |\gamma_i|^2\right] = \sum_{k=1}^{\eta}(1/k)$\cite{Lugo11}. By Theorem 1 and $\mathbb{E}\left[P^{\mathrm{keyhole}}_r\right] = N_t N_r \ell(d)$, $\Upsilon \approx \mathbb{E}\left[\max_i |\gamma_i|^2\right]/\eta = \frac{1}{\eta}\sum_{k=1}^{\eta}(1/k)$ if $\gamma_i$ are complex normal random variables and independent of each other. Similarly, the other two results are derived.  
\end{proof}

From Corollary 1, NLOS received signal power can be significantly overestimated with the keyhole model for $\eta = 10$, which translates to $\Upsilon = -4.6$dB if $\gamma_i$ are identically equal to $\gamma$, and to $\Upsilon = -10$dB if $\gamma_i$ are independently but identically distributed. Note that this is an analytical result and that well accepted wideband models (like in \cite{3gpp_tr38901}) will have unequal distribution of powers amongst paths within and across different clusters. Estimating $\Upsilon$ in these settings is an avenue for further research. 

\begin{thm}
Let the transmitter and receiver beamforming vectors be $\frac{{\bf a}_t(\theta')}{\sqrt{N_t}}$ and $\frac{{\bf a}_r(\phi')}{\sqrt{N_r}}$.
If $\gamma_i$ are independent zero mean random variables with unit variance,  $\{\theta_i\}$ are identically distributed to $\theta$,  $\{\phi_i\}$ are identically distributed to $\phi$ and $\{\gamma_i\}$ are independent of all AOAs/AODs, then $\mathbb{E}\left[P^{\mathrm{multi}}_r\right] = \mathbb{E}\left[P^{\mathrm{keyhole}}_r\right]$ for NLOS  interfering links.
\end{thm}
\begin{proof}
The received signal power considering a multipath channel in \eqref{eq:multipathchannel} is given by
\begin{equation}
\label{eq:intermulti}
P^{\mathrm{multi}}_r = \frac{\ell(d)}{\eta} \left|\left|\sum_{i=1}^{\eta}\frac{\gamma_i}{\sqrt{N_t N_r}}{\bf a}^*_r(\phi'){\bf a}_r(\phi_i){\bf a}^*_t(\theta_i){\bf a}_t(\theta')\right|\right|^2.
\end{equation} 

Using independence of $\gamma_i$ and that these are zero mean random variables, the cross terms while expanding the norm squared in \eqref{eq:intermulti} become zero and thus, 
\ifSC
$\mathbb{E}\left[P^{\mathrm{multi}}_r\right]$ is equal to
\begin{equation*}
\frac{\ell(d)}{\eta}\sum_{i=1}^{\eta}\mathbb{E}\left[|\gamma_i|^2\right]\mathbb{E}\left[G_r(\phi',\phi_i)G_t(\theta_i,\theta')\right] = \frac{\ell(d)\eta \mathbb{E}\left[G_r(\phi',\phi_1)G_t(\theta_1,\theta')\right]}{\eta}
 = \mathbb{E}\left[P^{\mathrm{keyhole}}_r\right].
\end{equation*}
\else
\begin{align*}
&\mathbb{E}\left[P^{\mathrm{multi}}_r\right]  = \frac{\ell(d)}{\eta}\sum_{i=1}^{\eta}\mathbb{E}\left[|\gamma_i|^2\right]\mathbb{E}\left[G_r(\phi',\phi_i)G_t(\theta_i,\theta')\right]\\
& = \frac{\ell(d)\eta \mathbb{E}\left[G_r(\phi',\phi_1)G_t(\theta_1,\theta')\right]}{\eta}
 = \mathbb{E}\left[P^{\mathrm{keyhole}}_r\right].
\end{align*}
\fi
\end{proof}
Theorem 2 indicates that a correction factor is not necessary for NLOS interfering links, if the assumptions in the theorem hold true. Similar result can be stated for LOS interfering links. However, depending on the structure of the arrays, the per-element antenna gains and joint distribution of $\{\gamma_i,\phi_i,\theta_i\}$ a non-unity correction factor may be necessary for interfering links. Next, we will validate the need for a correction factor with some simulations using the 3GPP NR channel model\cite{3gpp_tr38901}. 

\section{Simulation Result with 3GPP model}
We consider two MIMO systems with link lengths $100$ meters operating at 73 GHz carrier frequency. One is LOS and the other is NLOS. $8\times 8$ uniform planar array with half wavelength spacing is assumed at the transmitters and the receivers. Note that considering a $8\times 8$ antenna array system is realistic for mmWave backhaul networks wherein both ends of a communication link are base stations (BSs)\cite{Nokia16}. Effective antenna gain for each of these MIMO systems is computed as $P^\mathrm{multi}_r/\ell(d)$ as per Definition 2. Here, $P^\mathrm{multi}_r$ was computed considering the 3GPP NR channel model \cite{3gpp_tr38901} along with optimal precoders and combiners that maximize the SNR and a unit transmit power. Several realizations of the 3GPP channel were simulated for both the links. The distribution of effective antenna gain seen by the LOS and NLOS link is plotted in Fig.~\ref{fig:3}. As seen from Fig.~\ref{fig:3} there is a drop of about 12 dB in NLOS median gain compared to LOS, which is very significant. The implication of such  drop in effective antenna gain is discussed in next section. The LOS effective antenna gain in Fig. 1 is very close to $10\log_{10}(64\times 64) = 36$dB , as expected, since correction factor for LOS links is negligible as per our analysis.  Surprisingly the drop in NLOS gain is equal to $-10\log_{10}19$, wherein $19$ is the mean number of NLOS clusters in the 3GPP model. This equals our analytical estimate of $1/\eta$ considering $\eta = 19$. A more accurate analysis explicitly modeling different clusters with multiple rays and correlated small scale fading is a possible future work.

\ifSC
\else
\begin{figure}
\centering
\includegraphics[width = \columnwidth]{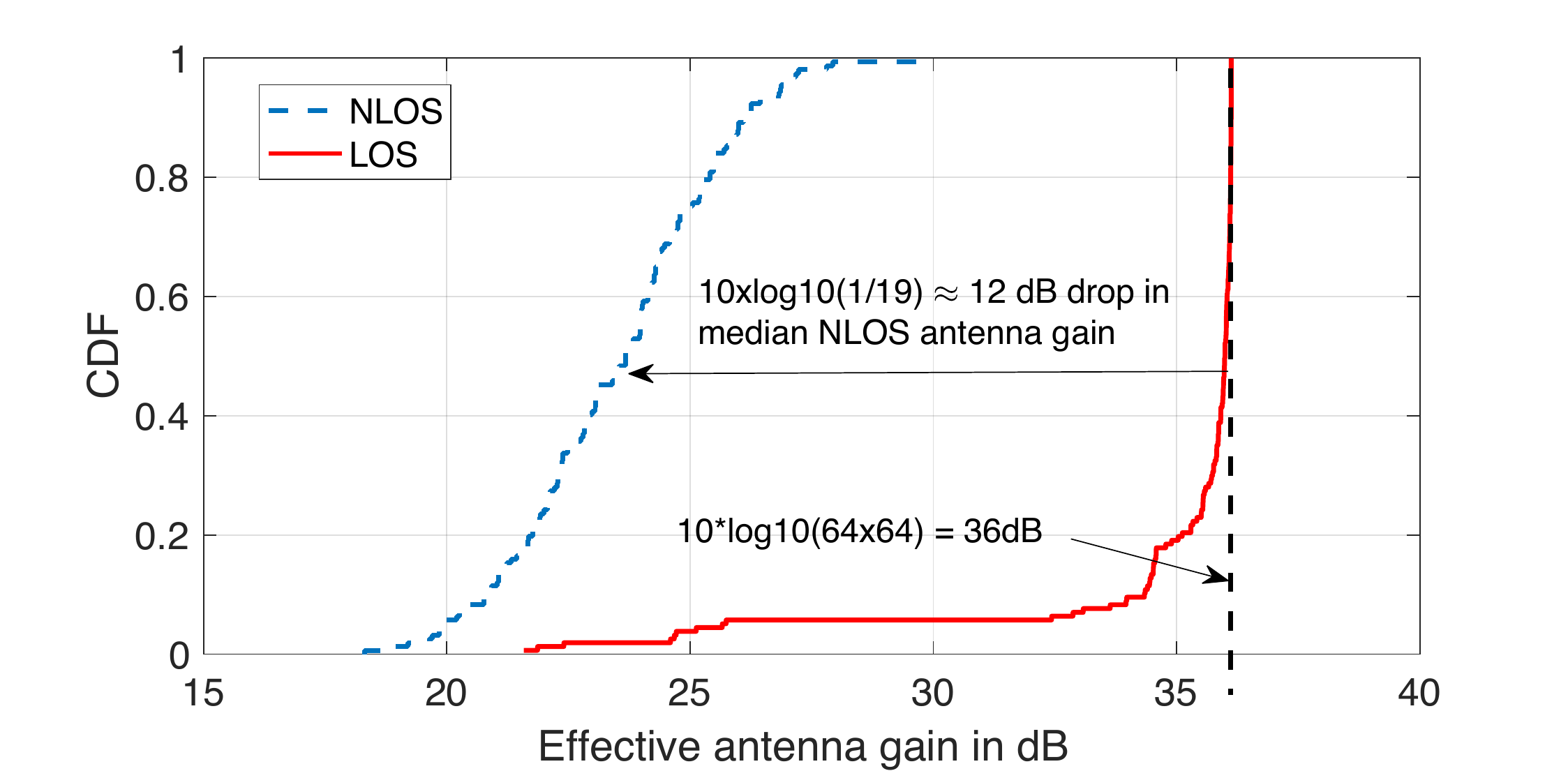}
\caption{Comparison of effective antenna gain for LOS and NLOS links with new radio 3GPP channel model\cite{3gpp_tr38901}.}
\label{fig:3}
\end{figure}
\fi
\section{Implications and Applicability of the Work}
This work is applicable for MIMO wireless networks with highly directional single stream beamforming at the transmitter and the receiver. The analysis can also be extended for multi-user MIMO with large number of transmit and receive antennas. In short, whenever the underlying signal processing of a large MIMO system is abstracted to compute the received signal power as a product of a single input single output (SISO) received signal power and some spatial antenna gain patterns at the transmitter/receiver for simplified analysis, there will be a need for incorporating the correction factor to make sure that identical antenna gain patterns are not multiplied for LOS and NLOS links, as well as serving and interfering links. The implications of the work are prominent in the following scenarios. For dense outdoor-to-outdoor cellular networks, a user would likely associate with a LOS BS and thus the signal to noise ratio (SNR) coverage estimates wouldn't vary significantly, except for the tail probability when a user associates with a NLOS BS that affects the cell edge rates. Otherwise, we expect such a correction factor to be significant since there is significant probability of connecting to a NLOS BS since the SNR distribution itself will shift by $\Upsilon$. We expect the significance of such a correction factor to also be significant in analysis of multi-hop mmWave cellular networks wherein the fiber site deployment will be relatively sparse and thus there will be a question as to whether a relay should go for a NLOS direct hop to fiber base station or whether it should relay over multiple LOS hops. Given that the correction factor introduced in this letter doesn't affect LOS links but strongly affects NLOS links, LOS hops will be even more strongly favoured over NLOS hops. Neglecting the correction factor but using a model like \eqref{eq:singlepath} can lead to misleading insights. 
\renewcommand{\appendixname}{Appendix: An Example Demonstrating The Use Of Correction Factor}
\begin{appendix}
Consider a receiving user at origin and a collection of transmitting base stations (BSs) in $\mathbb{R}^2$ including a BS at $Y\in\mathbb{R}^2$ which is NLOS with respect to the user. We want to understand the SINR performance of that user using the simplified received signal power in \eqref{eq:singlepath} that models beamforming through a spatial antenna pattern. Our proposal is to introduce the correction factor to compute the received signal powers. In principle, this factor is different for LOS and NLOS as well as for service and interfering links. For simplicity of exposition, we will consider the correction factor to be much less than 1 for NLOS serving links and equal to 1 for rest of the cases, which is an outcome of our asymptotic analysis. First to evaluate whether $Y$ is an interferer or a serving BS -- usually the serving BS is the one with maximum received signal power averaged over $h$ -- one has to multiply a correction factor that is much less than 1 to the received signal power from $Y$ to origin. However, if it is determined that the BS does not serve the user but is a potentially interfering BS, then the correction factor is equal to 1 while computing the interference power from the same BS at $Y$ to the receiver at origin. Such a modification in received power, which is done differently for service and interfering links cannot be done by modifying $\ell(.)$ or $h$.
\end{appendix}
\bibliographystyle{IEEEtran}
\bibliography{IEEEabrv,Kulkarni}
\end{document}